\documentclass[11pt]{article}

\usepackage[margin=1in]{geometry}
\usepackage[T1]{fontenc}
\usepackage[utf8]{inputenc}
\usepackage[expansion=false]{microtype}
\usepackage{amsmath,amssymb,amsthm,mathtools}
\usepackage{bm}
\usepackage{booktabs}
\usepackage{longtable}
\usepackage{array}
\usepackage{enumitem}
\usepackage{xcolor}
\usepackage{float}
\usepackage{hyperref}
\usepackage{authblk}
\usepackage[round,authoryear]{natbib}

\hypersetup{
  colorlinks=true,
  linkcolor=blue!55!black,
  citecolor=blue!55!black,
  urlcolor=blue!55!black,
  pdftitle={Intrinsic Computational Functionalism and Simulated Consciousness},
  pdfauthor={Ryota Kanai and Shuqin Ma}
}

\setlist[itemize]{leftmargin=1.5em}
\setlist[enumerate]{leftmargin=1.7em}

\newtheorem{definition}{Definition}
\newtheorem{principle}{Principle}
\newtheorem{assumption}{Assumption}
\newtheorem{theorem}{Theorem}
\newtheorem{proposition}{Proposition}

\newcommand{\K}{\mathfrak{K}}

\newcommand{\A}{\mathcal{A}}

\newcommand{\J}{\mathcal{J}}

\newcommand{\Sim}{\mathrm{Sim}}

\newcommand{\ICCR}{\mathrm{ICCR}}

\newcommand{\doop}{\mathrm{do}}
\newcommand{\Tr}{\mathrm{Tr}}
\newcommand{\Can}{\mathrm{Can}}
\newcommand{\MCan}{\mathrm{MCan}}
\newcommand{\eps}{\varepsilon}

\title{Intrinsic Computational Functionalism and Simulated Consciousness}

\author[1]{Ryota Kanai}
\author[2,3]{Shuqin Ma}
\affil[1]{Araya Inc., Tokyo, Japan}
\affil[2]{School of Philosophy, Fudan University, Shanghai, China}
\affil[3]{Sussex Centre for Consciousness Science, University of Sussex, Brighton, UK}

\date{}
\begin{document}
\maketitle

\begin{abstract}
A common objection to artificial or simulated consciousness is that a simulated brain is no more conscious than simulated water is wet. We address this from the perspective of Intrinsic Computational Functionalism (ICF): if consciousness is computationally constituted, it depends not on externally imposed descriptions but on the computational structures a system physically realizes in virtue of its own causal-dynamical organization. In previous work we developed Canonical Functionalism as a mathematically precise special case of this anti-interpretivist program, identifying functional states by their complete future input-output roles under a fixed interface. Here we argue that this input-output construction, though important, is incomplete: as a behavioral boundary case of ICF, it makes lookup tables and unfolded systems that preserve the same boundary behavior canonically equivalent. A consciousness-relevant canonical representation must instead include internal mechanisms, interventions, and joint readouts belonging to the relevant intrinsic organization. We therefore define a mechanism-enriched canonical structure and use it to formulate Intrinsic Causal-Computational Realization (ICCR), a realization relation preserving physical implementation, intrinsic state individuation, transition structure, intervention profiles, and the relevant agent-body-world boundary. The central result is conditional: if conscious properties are invariants of intrinsic causal-computational organization, then any system satisfying ICCR realizes the same consciousness-relevant properties, whether biological, artificial, or simulated. We discuss objections including biological naturalism and integrated information theory. We conclude that to deny consciousness to a simulation, one must identify a consciousness-relevant intrinsic causal-computational structure that the simulation fails to realize.

\end{abstract}


\clearpage
\section{Introduction}

One of the most persistent objections to artificial or simulated consciousness is the claim that a simulation is not the thing simulated. A simulated hurricane does not make the computer windy; a simulated furnace does not heat the room; a simulated stomach does not digest lunch; and a simulated water does not make the computer wet. By analogy, the objection continues, a simulated brain does not generate consciousness. It merely describes or imitates the dynamics of the brain.

The objection is powerful because it sounds like common sense. However, the common-sense formulation hides several distinctions. It does not distinguish external description from physical implementation, semantic interpretation from intrinsic organization, actual trajectory matching from counterfactual causal realization, or substrate identity from preservation of consciousness-relevant structure. In the context of consciousness science, these distinctions matter. If consciousness is grounded in a biological chemical substrate, then digital brain simulation may fail. If consciousness is grounded in electromagnetic fields, quantum processes, metabolic cycles, or some other non-computational physical property, then a digital simulation may again fail. But if consciousness is grounded in information-bearing causal organization, then the mere fact that a system is called a simulation by an external observer does not show that it lacks the relevant grounds.

This paper develops that response as a consequence of Intrinsic Computational Functionalism (ICF), a position recently formulated to rescue computational functionalism from observer-relativity and triviality objections \citep{MaKanai2026ICF}. The core idea of ICF is that, if consciousness is computationally constituted, it must be constituted by computational structures that are intrinsic properties of the physical system itself. That is, the relevant computational organization must not depend on how an external interpreter labels the system's states or assigns a task-relative meaning to its trajectories. ICF therefore imposes two central requirements. First, the relevant computation must be intrinsically instantiated for the system: it must be a property the system has in virtue of its own physical organization, stable state distinctions, and causal powers, rather than a description projected onto it from the outside. Second, it must be grounded in causal-dynamical organization under intervention: the system's states must stand in relations of mutual constraint and counterfactual dependence that can be revealed by perturbations and alternative inputs. In this view, a consciousness-relevant computational property is neither an arbitrary formal mapping nor a merely behavioral summary. It is an intrinsic causal-computational structure belonging to the system's own dynamics.

This intrinsic emphasis is also in line with the aspiration to develop a universal theory of consciousness \citep{KanaiFujisawa2024}. Universality refers to the ability of a theory to determine whether and in what way any fully described dynamical system is conscious. Such universality is precisely what is required to address the consciousness of artificial systems beyond the biological brain. Crucially, it has been argued that for a theory to be universal, the determinant of consciousness must be specified as an intrinsic property of the system rather than as something that depends on an external observer's interpretation. ICF inherits this commitment directly and insists on system-intrinsic instantiation and observer-independent causal-dynamical organization.

In this line of theoretical development, Canonical Functionalism \citep{KanaiMa2026Canonical} was our first attempt to define functional states of a system mathematically. It proposed that a functional state can be defined by its complete future behavior, i.e., how the system would respond under all possible future input histories. Mathematically, this yields a canonical quotient over states, analogous to minimal-state constructions in automata theory and computational mechanics. This was a useful step because it replaces arbitrary semantic labels with a mathematically determined structure. Once the input-output interface is fixed, states that have identical future input-output roles are identified.

The present paper treats Canonical Functionalism (CF) as a special, boundary-level case of ICF. As originally formulated, CF considers only the inputs to, and outputs from, the system as a whole. In this respect, it is analogous to a behavioral criterion: it characterizes the system by the functions it realizes at its external interface, rather than by the mechanisms through which those functions are produced. Consequently, if a lookup table, an unfolded system, and a recurrent system exhibit exactly the same complete boundary behavior, CF treats their canonical structures as isomorphic. This conclusion is logically coherent within the original framework, but this boundary-level equivalence likely misses the functional structural differences of the internal mechanisms, which are likely important for the realization of consciousness. 

This paper therefore extends the canonical construction to the internal mechanisms. The consciousness-relevant canonical object should not be defined only over the individual's external input-output profile. It should be defined over a family of interventions and readouts that include the internal mechanisms by which the system generates its behavior. Instead of asking only what output the individual would produce under each input history, we ask how every admissible internal and external intervention would change every admissible internal and external readout. The result is a mechanism-enriched canonical structure (or mechanistic canonical structure, for short). It preserves the anti-interpretivist motivation of Canonical Functionalism while avoiding the boundary-behavioral collapse that makes tables and unfolded systems too easily equivalent.

The paper's central claim is therefore as follows:
\begin{quote}
If consciousness depends on the intrinsic functional structure of a system, and if that structure is identified not only at the input-output boundary but also in the system's internal mechanisms, then any physical system that realizes the same relevant structure should preserve the same consciousness-relevant properties. Even if the system is implemented as a computer simulation from an external point of view, its status as a simulation does not by itself undermine its claim to consciousness. Provided that the relevant internal organization is actually realized, the system is not merely a description or model of consciousness, but a physical implementation of the structure on which consciousness depends.
\end{quote}

This claim deliberately avoids two overstatements. It does not claim that every computer simulation is conscious. It does not claim that present digital systems already instantiate the relevant organization. It also does not refute all substrate-dependent theories of consciousness. Rather, it shows that the simple inference from ``simulation'' to ``not conscious'' is invalid. A successful anti-simulation argument must identify which consciousness-relevant intrinsic structure is missing \citep{BaltieriKanai2025}.

The aim of the paper is not to derive artificial consciousness from computational functionalism alone. Rather, it is to clarify the realization relation that must hold if a simulated or artificial system is to share the consciousness-relevant organization of a biological system. The framework therefore shifts the debate from whether a system is called a simulation to whether it realizes, in its own physical dynamics, the relevant intrinsic causal-computational structure.

The paper proceeds as follows. Section \ref{sec:icf-cf} locates the proposal in the sequence from ICF to Canonical Functionalism and explains why a mechanism-enriched canonical object is needed. Section \ref{sec:taxonomy} distinguishes simulation, representation, implementation, realization, and duplication. Sections \ref{sec:formal}--\ref{sec:iccr} give the formal framework and define Intrinsic Causal-Computational Realization (ICCR). Section \ref{sec:theorems} states the conditional preservation results. Section \ref{sec:objections} groups and answers the main objections. Sections \ref{sec:doesnot}--\ref{sec:program} clarify the limits of the result and discuss how established theoretical constructs such as the global workspace may provide a practical route toward identifying the system-intrinsic, dynamics-internal structures on which consciousness depends.

\section{From ICF to Canonical Functionalism and beyond}\label{sec:icf-cf}

Computational functionalism about consciousness begins with the idea that mental states are individuated by their causal-functional roles rather than by their biological substrate. On this view, consciousness may be multiply realizable: the same consciousness-relevant organization could, in principle, be implemented in different physical substrates. Classical computational functionalism is often associated with machine-state functionalism, multiple realizability, and the claim that the correct kind of computation may be sufficient for consciousness.

The tradition faces well-known objections. Searle's Chinese Room argues that formal symbol manipulation is not sufficient for semantics or understanding \citep{Searle1980}. Block's absent-qualia objections ask whether functionally organized but intuitively bizarre systems would really have phenomenal experience \citep{Block1978}. Putnam-style triviality arguments threaten to show that, if implementation is merely a mapping from physical states to computational states, then every ordinary open physical system implements every finite automaton \citep{Putnam1988,Chalmers1996Rock,Sprevak2018}. Maudlin's Olympia presses a related problem about counterfactual computational structure and inactive machinery \citep{Maudlin1989,Barnes1991,Klein2008}. A more recent challenge comes from IIT-style anti-functionalism. IIT holds that consciousness depends not merely on input-output or computational equivalence, but on the intrinsic cause-effect power of a physical substrate \citep{Oizumi2014,TononiKoch2015,Albantakis2023}. In this spirit, \citet{Findlay2024} argue that a stored-program computer can simulate another Boolean system with full functional equivalence while failing to be phenomenally equivalent to it according to IIT. Their result provides a sharp contemporary version of the simulation objection: even perfect functional simulation, on this view, need not preserve consciousness unless the relevant intrinsic physical cause-effect structure is also preserved.

ICF accepts these objections against naive computationalism. The central claim of ICF is that a consciousness-relevant computational property must be an intrinsic property of the physical system itself. ICF therefore gives priority to two operational criteria: system-intrinsic instantiation (C1) and causal-dynamical organization under intervention (C2) \citep{MaKanai2026ICF}. C1 requires that the relevant computational structure be grounded in distinctions the system has in virtue of its own physical organization, rather than in labels projected by an interpreter. C2 requires that these distinctions be embedded in the system's own causal dynamics, so that the relevant variables mutually constrain one another and their organization is revealed by perturbations and alternative inputs. The role of empirical theory is not to impose the computation from the outside, but to help identify the internal functional structure that the system realizes. The present paper applies this idea to the simulation objection by asking what must be preserved for a simulated or artificial system to realize the same intrinsic functional organization as a conscious biological system.

We now make this anti-interpretivist idea mathematically precise, following Canonical Functionalism \citep{KanaiMa2026Canonical}. For an interactive deterministic system
\[
S=(X,x_0,\delta,o),
\]
with input set \(I\), output set \(O\), transition function \(\delta:X\times I\to X\), and output function \(o:X\to O\), the future behavior of state \(x\) is
\[
b_x^S(w)=o(\delta^*(x,w)),\qquad w\in I^*.
\]
States are equivalent when they have identical future behavior under every possible input history:
\[
x\sim_S y \quad \text{iff}\quad \forall w\in I^*,\ b_x^S(w)=b_y^S(w).
\]
The canonical functional structure is the quotient of the reachable state space by this equivalence relation:
\[
\Can(S)=R_S/\sim_S.
\]
The canonical realization theorem states that any two systems with the same complete input-output behavior determine isomorphic canonical structures. This is the sense in which Canonical Functionalism removes observer-relative semantics: once the interface is fixed, the quotient is fixed by counterfactual behavior, not by the theorist's labels.

However, it is important to note that this construction only captures a part of the ICF since it is canonical only relative to the input and the output of the whole system and ignores the internal structure which should be also captured as the intrinsic property of the system. The canonical functional structure as described above identifies internal states only through their complete future roles with respect to the input-output interface. That is why it does not distinguish lookup tables and unfolding cases from other equivalent, possibly more intricate systems. If a lookup table is expanded so that it contains the residual behavior for every possible input history, then at the boundary-canonical level it realizes the same canonical machine. If an unfolded system preserves the same complete future behavior as a recurrent system, then \(\Can(R)\cong\Can(F)\). These are the consequences of making the external interface the sole criterion of individuation.

For consciousness, however, this boundary-canonical structure is not sufficient. A theory of consciousness may claim that various forms of internal causal structures such as recurrence, global availability, higher-order monitoring, affective regulation, temporally continuous self-conditioning, or integrated cause-effect power are important for consciousness. If so, the canonical representation relevant to consciousness cannot be only the minimal boundary machine. It must include internal interventions and internal readouts. For example, it must ask how perturbing a recurrent perceptual loop changes later perceptual integration, how disrupting a candidate workspace changes access by specialized modules, how altering a self-model variable changes confidence, memory, and action, or how changing precision-weighting variables reorganizes downstream control. The point is to formulate the realization relation so that such theories have room to specify what must be preserved.

The present paper therefore introduces a mechanism-enriched canonical construction. It retains the ICF demand that the structure be intrinsic and intervention-sensitive, but it no longer treats the individual as a black box with only one external input-output interface. It generalizes the canonical role from boundary behavior to the full family of theory-admissible intervention-readout relations. Thus, canonical Functionalism is a special case of this more generalized framework. 

\section{Simulation, implementation, realization, duplication}\label{sec:taxonomy}

The simulation objection depends on distinctions that are often left implicit. To say that one system simulates another is not yet to say that it implements the same computation, and to say that it implements a computation is not yet to say that the relevant structure is intrinsic to the system itself. These distinctions are important, as the present paper does not defend arbitrary simulation or behavioral matching. Its target is a stronger relation: the realization of an intrinsic functional structure that can later be characterized in a mechanism-enriched canonical form. The definitions in this section therefore separate representation, simulation, implementation, and intrinsic realization. This separation allows us to ask, in later sections, whether a simulated or artificial system preserves not merely the same input-output behavior, but the same internally organized structure relevant to consciousness.

\begin{definition}[External simulation]
A physical system \(M\) externally simulates a target system \(B\) relative to a modeling relation \(R\) when states, variables, and outputs of \(M\) are interpreted by an external observer as representing states, variables, and outputs of \(B\). We write
\[
\Sim_R(M,B).
\]
The relation is observer-relative because \(R\) may depend on semantic conventions, measurement choices, encoding schemes, or modeling goals.
\end{definition}

A climate model externally simulates a storm because model variables are interpreted as pressure, temperature, and velocity fields. A neural simulation externally simulates a brain because memory states and numerical arrays are interpreted as membrane potentials, synaptic weights, neurotransmitter concentrations, or spike trains.

\begin{definition}[Implementation]
A physical system \(M\) implements an abstract computation \(C\) when there is a correspondence between physical states of \(M\) and computational states of \(C\) that preserves transitions: whenever \(C\) would move from one computational state to the next on a given input, the corresponding physical states of \(M\) evolve accordingly under the physical realization of that input. This must hold across the range of inputs and perturbations over which \(C\) is defined, not merely along a single realized run. The last clause rules out a cheap trick: reading a correspondence off one recorded trajectory after the fact, which matches \(C\) only on the run that actually happened.
\end{definition}

Implementation is stronger than external simulation. It is not enough that an observer can describe \(M\) as if it were computing \(C\); the system must have physical states and transitions that actually carry the computation. This is the lesson of the implementation literature: if the correspondence is left unconstrained, implementation becomes trivial, because a clever relabelling can map almost any system onto almost any computation \citep{Chalmers1994,Chalmers1996Rock,Piccinini2015,Sprevak2018}. But implementation still leaves one thing open. It asks only whether \emph{some} transition-preserving correspondence exists. It does not ask where the states in that correspondence come from---whether they are real divisions in the system, or groupings an observer chose. That is the further question intrinsic realization adds.

\begin{definition}[Intrinsic realization]
A physical system \(M\) intrinsically realizes a computational structure \(C\) when \(M\) implements \(C\), and the states used in the correspondence are the system's own. That is, each group of physical states treated as a single computational state must be a real, causally effective distinction in \(M\): something that makes a difference to how \(M\) behaves and responds when perturbed, rather than a label an observer has placed on it. So implementation asks whether a transition-preserving correspondence exists; realization asks, in addition, whether its states are the system's own. When they are, \(C\) is not just a description that fits \(M\); it is a structure \(M\) has in virtue of itself.
\end{definition}

Realization, not implementation, is what matters for consciousness, and the gap between the two is where the simulation debate is settled. The clearest case is one where internal mechanism is not the issue at all. Take a system with a rich set of internal states, and let an observer cleverly code those states so that they track the transitions of some computation \(C\). If the coding keeps working across the relevant inputs, the system implements \(C\). But the divisions doing the work belong to the observer's code, not to the system: perturb the system, and the coded ``states'' need not behave like states of \(C\) at all. So the system implements \(C\) without intrinsically realizing it. This is the loophole that Putnam-style triviality arguments exploit (Section~\ref{sec:objections}), and it is exactly what intrinsic realization shuts: it demands that the states come from the system's own dynamics. Note that this is a question about where the states come from, not about whether we read out internal parts or only the system's outputs---that second question comes later.

ICF holds that consciousness-relevant computation must be an intrinsic organization the system realizes in virtue of its own dynamics, not just a computation an observer can map onto it. In the vocabulary of the mechanistic account of computation, intrinsic realization asks for medium-independent vehicles individuated by the system's own functional mechanisms rather than by external semantics; ICCR can be read as a counterfactual, intervention-based sharpening of that mechanistic requirement \citep{Piccinini2015}. In this paper, we formulate a canonical construction for this intrinsic realization. It aims to identify the internal organization the system itself realizes, so that systems can be compared by their consciousness-relevant intrinsic structure rather than by their boundary input-output behavior alone.

\begin{definition}[Duplication]
A physical system \(M\) duplicates another physical system \(B\) when \(M\) reproduces not only the relevant functional or computational organization of \(B\), but also the relevant physical constitution of \(B\) at the level of description in question. Duplication is therefore stronger than implementation or intrinsic realization. A duplicate of a biological brain would reproduce the brain's biological or microphysical organization, whereas a non-biological system may realize the same consciousness-relevant organization without duplicating the biological substrate.
\end{definition}

This distinction is important because the simulation objection often treats the failure of duplication as if it were a failure of realization. A computer simulation of a brain is not a biological duplicate of that brain, just as simulated water is not a duplicate of ordinary water. But ICF does not require duplication of the original substrate. It requires realization of the intrinsic organization on which consciousness depends. Thus, the relevant question is not whether an artificial system is made of the same biological material as the brain, but whether it realizes the same consciousness-relevant intrinsic structure.

These distinctions allow us to diagnose the water analogy. A simulated water system need not duplicate external-world \(\mathrm{H_2O}\). It may nevertheless realize, inside its implemented domain, the causal roles that make a virtual body float, sink, dissolve salt, or register cooling. The question is not whether the external computer is wet. The question is what domain of causal relations is under consideration and which level of realization matters.

Among these distinctions, ICF claims that only the realization is the relevant level for consciousness, and mere simulation is not sufficient, and duplication is unnecessary. A system does not become conscious simply by being described as running the right computation, nor must it duplicate the biological substrate of the brain in every physical detail. Rather, if consciousness is computationally constituted, it arises when the relevant computational organization is intrinsically realized by a physical system. The central question is therefore whether an artificial or simulated system realizes, in its own physical dynamics, the same consciousness-relevant organization as a biological brain. The later sections develop a canonical way of formulating this realization relation so that the comparison is not limited to boundary input-output behavior but extends to the internal mechanisms that generate it.

\section{Formal framework}\label{sec:formal}

This section gives a mathematical framework for the argument. While the formalization is aimed to be generalizable, the constructions below are carried out for the discrete-time, finite-state, deterministic case. Extending them to continuous time and to genuinely stochastic dynamics raises further formal questions, and we do not claim to have established them here.

\subsection{Physical systems as interventional dynamical systems}

\begin{definition}[Interventional physical system]
An interventional physical system is a tuple
\[
P=(X,\Omega,\A,\J,\{\eta_J\}_{J\in\J},\Phi).
\]
Here \(X\) is the physical state space of the system. The symbol \(\Omega\) denotes the set of admissible input or context continuations, representing the possible ways in which the system may continue to receive input from its environment. The symbol \(\A\) denotes the set of admissible interventions, including perturbations, stimulations, lesions, clamping operations, or other manipulations that can be applied to the system. The symbol \(\J\) denotes a family of admissible readout sets, specifying which parts or variables of the system are jointly read out. For each admissible readout set \(J\), the map \(\eta_J:X\to Y_J\) gives the joint readout associated with \(J\).

The symbol \(\Phi\) denotes the system's time-evolution rule. More specifically, for each time \(t\), each intervention \(a\) in \(\A\), and each input or context continuation \(\omega\) in \(\Omega\), the map
\[
\Phi_t^{a,\omega}:X\to X
\]
gives the evolution of the system under that intervention and input or context continuation. Given an initial state \(x\) in \(X\), the expression \(\Phi_t^{a,\omega}(x)\) denotes the state reached after time \(t\) from \(x\). When stochasticity is essential, \(\Phi_t^{a,\omega}(x)\) may denote a transition kernel over \(X\) rather than a deterministic next state.
\end{definition}

For a biological agent, \(X\) may include neural, bodily, and environmental variables. For a digital machine, \(X\) includes registers, memory, clocks, buses, sensors, actuators, power state, and relevant environmental couplings. For an artificial agent in a virtual world, \(X\) may include both the lower-level hardware state and the virtual agent-world state, provided that the latter is physically realized by the former. The readout family \(\J\) is what distinguishes the present framework from the boundary-canonical construction. It may include an external behavioral output, but it may also include internal variables, modules, recurrent pathways, broadcast states, confidence variables, memory traces, affective control variables, bodily variables, and joint readouts over combinations of such components.

In ICF, the formal target of the intervention-readout organization must be generated by the system itself, i.e., the admissible interventions and readouts should be constrained by the system's own physical organization.

\subsection{Mechanism-enriched counterfactual roles}

Given \(x\in X\), \(a\in\A\), \(\omega\in\Omega\), and readout set \(J\in\J\), define the readout trajectory
\[
\Tr^J_P(x,a,\omega)=\left(\eta_J(\Phi_t^{a,\omega}(x))\right)_{t\geq 0}.
\]
Two states may be identical along the actually realized behavioral trajectory but differ in what they would do under alternative internal interventions or alternative internal readouts. ICF treats the latter counterfactual profile as essential.

\begin{definition}[Mechanism-enriched counterfactual role]
The mechanism-enriched counterfactual role of state \(x\in X\) in system \(P\), relative to the chosen intervention and readout families, is the function
\[
\rho^\J_P(x):\A\times\Omega\times\J\to \bigcup_{J\in\J}Y_J^{\mathbb{T}},
\qquad
\rho^\J_P(x)(a,\omega,J)=\Tr^J_P(x,a,\omega),
\]
where \(\mathbb{T}\) is the relevant time index set.
\end{definition}

This definition extends the canonical functionalist idea that a state is individuated by its complete future behavior under all possible continuations \citep{KanaiMa2026Canonical}. The difference is that the continuations now include internal interventions and the readouts include internal mechanisms. The object is no longer only a boundary input-output profile. It is a profile over all admissible combinations of intervention and readout relations.

\subsection{Boundary canonical structure as a special case}

The earlier canonical quotient is recovered when the only readout is the external output and the only interventions are ordinary input histories.

\begin{definition}[Boundary canonical equivalence]
Suppose \(\J=\{J_{\mathrm{out}}\}\), \(\eta_{J_{\mathrm{out}}}=o\), and \(\A\) contains only the interventions corresponding to ordinary input histories. Then define
\[
x\sim^{\mathrm{bd}}_P x' \quad \text{iff}\quad \rho^\J_P(x)=\rho^\J_P(x').
\]
The boundary canonical structure is
\[
\Can(P)=R_P/\sim^{\mathrm{bd}}_P,
\]
where \(R_P\) is the reachable state set.
\end{definition}

\begin{proposition}[Canonical Functionalism as the boundary-output special case]
For a Moore-type interactive system \(S=(X,x_0,\delta,o)\), the boundary canonical structure defined above is isomorphic to the canonical functional structure \(\Can(S)=R_S/\sim_S\) as described in \citep{KanaiMa2026Canonical}.
\end{proposition}

\begin{proof}
In the boundary case, \(\rho^\J_P(x)\) records exactly the output trajectory generated by applying each possible future input history to \(x\). Hence \(\rho^\J_P(x)=\rho^\J_P(y)\) iff \(b_x^S(w)=b_y^S(w)\) for all \(w\in I^*\), which is precisely \(x\sim_S y\). The quotient structures are therefore isomorphic. \qedhere
\end{proof}

This proposition explains both the strength and the limitation of Canonical Functionalism. It is observer-independent once the interface is fixed, but it is interface-behavioral. It canonicalizes the individual as an input-output system. If two mechanisms differ internally while preserving every boundary continuation, the boundary quotient identifies them. This is why the lookup-table and unfolding conclusions follow. To represent consciousness-relevant mechanism, we must enrich the canonical object.

\subsection{Mechanism-enriched canonical structure}

\begin{definition}[Mechanism-enriched equivalence]\label{def:me-equivalence}
For a system \(P\), define
\[
x\sim^{\mathrm{me}}_P x' \quad \text{iff}\quad
\forall a\in\A,\ \forall \omega\in\Omega,\ \forall J\in\J,
\ \rho^{\J}_P(x)(a,\omega,J)=\rho^{\J}_P(x')(a,\omega,J).
\]
The mechanism-enriched canonical structure is
\[
\MCan(P)=R_P/{\sim^{\mathrm{me}}_P},
\]
with the induced transition, intervention, and readout structure inherited from \(P\).
\end{definition}

This definition is constitutively theory-independent. It does not say that two states are equivalent because a scientific theory chooses to treat them as equivalent. It says that they are equivalent when their full admissible intervention-readout roles coincide within the system's own dynamical organization. The qualifier ``admissible'' is therefore important: it marks interventions and readouts that belong to the physical organization and dynamical regime of the system, not arbitrary semantic projections imposed by an observer. In practice we may only approximate this complete object, but the mathematical target is not a theory-relative construction.

The mechanism-enriched quotient changes the status of the table and unfolding cases. A recurrent system \(R\) and an unfolded system \(F\) may satisfy \(\Can(R)\cong\Can(F)\) at the boundary. But once admissible interventions on recurrent pathways and readouts of their downstream effects are included, \(\MCan(R)\) and \(\MCan(F)\) come apart; Section~\ref{sec:toy} works through an explicit case in which no structure-preserving bijection of the intervention-readout families exists. Likewise, a lookup table that merely stores boundary residuals may realize \(\Can(S)\) while failing to realize \(\MCan(S)\). A table could pass only if its physical organization supports the same mechanism-enriched intervention-readout profile. At that point it is no longer a mere behavioral table in the objectionable sense; it is a physical mechanism realizing the target structure.

It is useful to locate the mechanism-enriched quotient relative to computational mechanics. The boundary quotient \(\Can(P)\) is close to the causal-state construction of an \(\eps\)-machine: states are identified when they induce the same conditional distribution over futures \citep{ShaliziCrutchfield2001,Crutchfield2012}. The mechanism-enriched quotient \(\MCan(P)\) generalizes this in two respects. First, it conditions not only on input-output futures but on the joint trajectories of internal readouts. Second, and more consequentially, it conditions on \emph{interventional} rather than merely observational continuations: two states are distinguished when they would respond differently to some admissible \(\doop(\cdot)\), even if they are statistically indistinguishable under passive observation. The minimality and uniqueness results of computational mechanics therefore transfer to the boundary case but not directly to \(\MCan\); recovering analogous results in the interventional, mechanism-enriched setting is an open formal question, and is tied to the grain-selection problem of Section~\ref{sec:grain}.

\subsection{Intrinsic partitions}

The mechanism-enriched quotient presupposes a partition of the physical state space. ICF requires that such a partition be intrinsic to the system rather than imposed by an interpreter.

\begin{definition}[Intrinsic partition]\label{def:intrinsic-partition}
A partition \(\Pi_P:X\to S\) is \emph{intrinsic} for \(P\) when it satisfies the
following conditions:
\begin{enumerate}[label=\textup{(\arabic*)},ref=I\arabic*,leftmargin=*,itemsep=3pt]
  \item\label{ip:support} \textbf{Physical support.} Transitions among cells are
        supported by physical causal processes within the system boundary.
  \item\label{ip:intervention} \textbf{Intervention tracking.} Admissible interventions produce systematic, law-governed changes in partition-state transitions and in admissible readouts.
    \item\label{ip:relabel} \textbf{Relabelling invariance.} The induced structure is
        invariant under structure-preserving relabellings of the system's variables.
  \item\label{ip:role} \textbf{Causal-dynamical individuation.} The partition states are individuated by causal-dynamical role rather than by arbitrary semantic interpretation.
\item\label{ip:lumpable} \textbf{Dynamical closure (lumpability).} The partition
        is dynamically closed under the system's evolution, in the sense that the
        cell \(\Pi_P(x)\), the applied intervention \(a\), and the input
        continuation \(\omega\) jointly determine the successor cell
        \(\Pi_P(\Phi^{a,\omega}(x))\), so that the induced transition of
        Definition~\ref{def:icc-structure} is well defined.
\end{enumerate}
\end{definition}

Intrinsicality so defined is a property of the partition relative to the system's dynamics alone. Whether a given intrinsic partition is also \emph{consciousness-relevant} is a further question, taken up in Sections~\ref{sec:grain} and~\ref{sec:program}.

The definition intentionally avoids making the partition relative to a prior theory \(T\). A theory may help us discover or approximate an intrinsic partition, but it does not constitute the partition as intrinsic. This preserves the C1 commitment of ICF: the consciousness-relevant structure must belong to the system itself. It also preserves C2: the relevant distinctions must be expressed in the system's own counterfactual responses under admissible inputs and interventions. A movie of a brain or a rock mapped by a fanciful coding scheme should fail not because we dislike the conclusion, but because the proposed partition lacks intervention tracking, and physically supported transitions.




\subsection{Intrinsic causal-computational structure}
Given an intrinsic partition \(\Pi_P:X\to S_P\), we can define the system's causal-computational organization.

\begin{definition}[Intrinsic causal-computational structure]\label{def:icc-structure}
The intrinsic causal-computational structure of \(P\) at partition \(\Pi_P\) is
\[
\K(P,\Pi_P)=(S_P,\Omega,\A,\J,\Delta_P,\Lambda_P),
\]
where \(S_P\) is the intrinsic state space, \(\Omega\) is the admissible input/context family, \(\A\) is the admissible intervention family, \(\J\) is the admissible readout family, \(\Delta_P^{a,\omega}:S_P\to S_P\) is the induced transition structure satisfying
\[
\Delta_P^{a,\omega}(\Pi_P(x))=\Pi_P(\Phi^{a,\omega}(x)),
\]
and \(\Lambda_P:S_P\times\A\times\Omega\times\J\to \bigcup_{J\in\J}Y_J^{\mathbb{T}}\) is the induced mechanism-enriched counterfactual response profile.
\end{definition}

This object is the formal target of ICF in the present paper. It is not a purely abstract computation floating free of physics, and it is not a structure made intrinsic by a theory. It is the quotient of a physical system under an intrinsic, intervention-sensitive, mechanism-enriched partition. The two constructions introduced above are related as follows: when the partition is the one induced by the mechanism-enriched equivalence \(\sim^{\mathrm{me}}_P\) of Definition~\ref{def:me-equivalence}---the finest partition that the admissible families \((\A,\Omega,\J)\) can resolve---the structure \(\K(P,\Pi_P)\) coincides with the mechanism-enriched canonical structure \(\MCan(P)\). For coarser intrinsic grains \(\K(P,\Pi_P)\) is the appropriate object, and the realization relation defined in the next section is accordingly stated at a fixed but otherwise arbitrary intrinsic grain. Later, in the discussion, we return to the practical question of how empirical theoretical constructs can guide attempts to approximate this dynamics-internal target.

\subsection{A toy model: recurrent realization versus boundary table}\label{sec:toy}

A simple example clarifies the difference between boundary equivalence and mechanism-enriched realization. Consider a target system \(B\) with two binary internal variables \(z_t=(x_t,y_t)\), a binary input \(u_t\), and external output \(o_t=x_t\). Its update rule is
\[
x_{t+1}=y_t,\qquad y_{t+1}=x_t\oplus u_t,
\]
where \(\oplus\) denotes exclusive-or. At the boundary, the target is characterized by the output sequence generated by each input history. Internally, however, the target also has a two-variable recurrent mechanism. Let the admissible interventions include \(\doop(x_t=c)\), \(\doop(y_t=c)\), and \(\doop(u_t=c)\), and let the admissible readouts include \(J_x,J_y,J_o,J_{xy}\), with \(c\in\{0,1\}\). Then the relevant structure is not just the actual output sequence; it is the profile of how the two variables and their joint state would evolve under alternative inputs and interventions.

Now compare three candidate systems. The first, \(M_R\), is a recurrent virtual machine with two physically realized registers \(r_x,r_y\), a physically realized input register \(r_u\), and an update mechanism that writes
\[
r_x' = r_y,\qquad r_y'=r_x\oplus r_u.
\]
The registers may be silicon flip-flops, memory cells under clocked control, or another stable physical medium. What matters is that the register partition is dynamically stable, writes and reads are physically causally supported, and interventions on \(r_x\), \(r_y\), or \(r_u\) change subsequent evolution in the way specified by the target. Under the partition mapping physical register states to \((x,y,u)\), \(M_R\) preserves the mechanism-enriched transition and intervention-readout profile of \(B\). It is therefore an ICCR realizer of this toy target, leaving aside any claim that the toy target is conscious.

The second, \(M_F\), is an \emph{unfolded} (feedforward) realizer. Fix a horizon \(T\) and lay out the recurrence as an acyclic circuit over distinct physical nodes \(x_0,y_0,x_1,y_1,\ldots,x_T,y_T\), wired so that
\[
x_{t+1}=y_t,\qquad y_{t+1}=x_t\oplus u_t,
\]
where each node is computed once by a feedforward gate and never overwritten. For any input stream of length at most \(T\), \(M_F\) produces exactly the output sequence of \(B\). This is the unfolding construction familiar from the unfolding argument against causal-structure theories \citep{Doerig2019,HansonWalker2019}: at the boundary, \(\Can(M_F)\cong\Can(M_R)\cong\Can(B)\), so no input-output test can separate them.

The mechanism-enriched object behaves differently. In \(M_R\) a single persistent component \(r_y\) carries the role of \(y\) at \emph{every} time step: it can be read out by \(\eta_{J_y}\) at any \(t\), and the localized intervention \(\doop(r_y:=c)\) acts on that one component and propagates through the loop, so that clamping it at time \(t\) reshapes the entire subsequent trajectory of the same two-component state. In \(M_F\) there is no single component that plays the role of \(y\) across time; the role is spread over the distinct nodes \(y_0,\ldots,y_T\). A readout ``the current value of the \(y\)-component'' is realized by one intrinsic variable in \(M_R\) but by no intrinsic variable in \(M_F\), and a \emph{sustained} single-component clamp \(\doop(r_y:=c)\) in \(M_R\) has no single-node counterpart in \(M_F\): mimicking it requires the coordinated multi-node intervention \(\{\doop(y_t:=c)\}_{t}\). Hence there is no structure-preserving bijection \(h_\A,h_\J\) between the admissible intervention-readout families of \(M_R\) and \(M_F\), and
\[
\MCan(M_R)\not\cong\MCan(M_F).
\]
If the consciousness-relevant target structure of \(B\) includes the persistent recurrent component---which a recurrence- or integration-based theory will assert, but a purely boundary-behavioral one will not---then \(M_F\) fails ICCR for \(B\) even though it passes every boundary-canonical test. This is precisely where the unfolding argument loses its grip: it establishes boundary equivalence, not mechanism-enriched equivalence. It therefore tells against input-output functionalism, and against causal-structure theories assessed on actual dynamics alone, but not against ICCR, whose verdict turns on the intervention-readout algebra rather than on the realized trajectory. Whether the persistent component is in fact consciousness-relevant is a tier-II/III question the framework deliberately leaves open.

The third, \(M_B\), is a boundary table. It stores, for every input history, the external output \(o_t\) that \(B\) would produce. At the boundary-canonical level, an ideal version of \(M_B\) can realize the same \(\Can(B)\), because it can reproduce the same residual input-output behavior. But \(M_B\) does not thereby realize \(\MCan(B)\). A perturbation that sets \(x_t=1\) in the target has no law-governed counterpart in the table unless the table is physically organized to represent and update \(x\) as an intervention-sensitive variable. A readout of \(y_t\) has no intrinsic counterpart unless the table contains a stable state component whose causal role matches \(y\). Any mapping from hidden addresses of the table to \(x\) and \(y\) is imposed by the interpreter unless it is supported by the table's own transition organization. Thus \(M_B\) may pass the boundary-canonical test while failing the mechanism-enriched test.

A trace player fails even more obviously. For one chosen initial state and one chosen input stream \(u^*\), it stores the output sequence \(o^*_0,o^*_1,o^*_2,\ldots\) and displays \(o^*_t\) at time \(t\). Along that one history, it can be output-indistinguishable from \(B\). But it does not match the counterfactual role function even at the boundary, let alone the mechanism-enriched one. The contrast is therefore not ``recurrent equations good, tables bad.'' It is realized counterfactual organization versus externally interpreted trajectory or boundary matching.

What this toy model establishes is the \emph{discriminating power} of the mechanism-enriched view: exact mechanism-enriched equivalence already separates recurrent realizers from lookup tables and unfolded systems, a distinction that boundary-canonical equivalence cannot draw. It does not yet speak to whether the relation is \emph{satisfiable} by realistic systems, for which exact equivalence is too demanding; that is the separate question of approximate realization addressed in Section~\ref{sec:approx}.

\section{Intrinsic Causal-Computational Realization}\label{sec:iccr}

We now define when one system realizes the consciousness-relevant organization of another. Let \(B\) be a biological conscious system, or a broader agent-body-world system containing a biological brain. Let \(M\) be a candidate artificial, simulated, neuromorphic, digital, analog, or hybrid system.

\begin{definition}[ICCR]
\(M\) intrinsically causally-computationally realizes \(B\), written
\[
M \models_{\ICCR} B,
\]
when there exist intrinsic partitions \(\Pi_B:X_B\to S_B\), \(\Pi_M:X_M\to S_M\), and structure-preserving bijections
\[
h:S_B\to S_M,\qquad h_\Omega:\Omega_B\to\Omega_M,
\qquad h_\A:\A_B\to\A_M,\qquad h_\J:\J_B\to\J_M,
\]
that satisfy physical realization, transition preservation, counterfactual preservation, and boundary adequacy.
\end{definition}

The conditions can be stated compactly. The definition already requires both partitions \(\Pi_B,\Pi_M\) to be intrinsic in the sense of Definition~\ref{def:intrinsic-partition}; this is what prevents the relation from being satisfied by an observer's arbitrary relabelling of states. Physical realization requires \(M\) to be a real physical system with states and transitions causally instantiated, not merely described. Transition preservation requires that, for all relevant \(s\in S_B\), \(a\in\A_B\), and \(\omega\in\Omega_B\),
\[
h\left(\Delta_B^{a,\omega}(s)\right)=
\Delta_M^{h_\A(a),h_\Omega(\omega)}(h(s)).
\]
Counterfactual preservation requires that, for all relevant \(s,a,\omega,J\),
\[
h\left(\Lambda_B(s,a,\omega,J)\right)=
\Lambda_M(h(s),h_\A(a),h_\Omega(\omega),h_\J(J)),
\]
with \(h\) applied pointwise to state trajectories and the appropriate readout isomorphism applied to readout trajectories. Because the partitions \(\Pi_B,\Pi_M\) are required to be intrinsic (Definition~\ref{def:intrinsic-partition}), the components carrying these preserved counterfactual branches must be ones that are load-bearing in the system's actual dynamics, rather than disconnected or merely appended machinery that an observer could describe but that plays no part in producing the current state. Boundary adequacy requires the realized system to include whatever brain, body, environmental, temporal, thermodynamic, or social variables fall within the system boundary and admissible families that fix the target organization \(\K(B,\Pi_B)\). In other words, the adequate boundary is the one already specified by the admissible interface for \(B\), not an independently stipulated notion: if a variable is part of the interface that individuates the target structure, the realizer must include its counterpart.

ICCR is intentionally stronger than ordinary input-output equivalence. It is also stronger than actual-trajectory matching and stronger than the earlier boundary-canonical quotient. A movie of a brain fails transition and counterfactual preservation. A static file containing a connectome fails physical realization of the relevant dynamics. A boundary lookup table may reproduce external behavior for every input history, but unless its internal organization realizes the relevant mechanism-enriched intervention-readout profile, it fails ICCR. A stored-program digital computer may or may not satisfy ICCR depending on whether its hardware-level and virtual-level organization realize the relevant intrinsic partition and intervention structure. That is a substantive scientific and philosophical question, not something settled by the word ``simulation.''

\subsection{ICCR and simulation}

A system can satisfy both \(\Sim_R(M,B)\) and \(M\models_{\ICCR}B\). There is no contradiction. \(\Sim_R(M,B)\) is a relation between \(M\), \(B\), and an external modeling relation \(R\). \(M\models_{\ICCR}B\) is a relation between intrinsic causal-computational structures. The former is observer-relative; the latter is intended to be observer-independent up to structure-preserving isomorphism.

This point is the core response to the simulation objection. A simulated brain is not a biological brain in the external substrate. But it may be a realizer of the same consciousness-relevant intrinsic organization. Whether it is depends on ICCR, not on the label ``simulation.''

\section{Conditional preservation theorems}\label{sec:theorems}

The formal payoff is a conditional preservation result.

\begin{principle}[ICF invariance]\label{pr:invariance}
Let \(Q\) be a consciousness-relevant property. If \(Q\) is grounded in intrinsic causal-computational organization, then there exists a function or rule \(F_Q\) such that
\[
Q(P)=F_Q(\K(P,\Pi_P)),
\]
and \(F_Q\) is invariant under isomorphisms of intrinsic causal-computational structure.
\end{principle}

This principle states the ICF commitment in mathematical form. It does not assert that consciousness is grounded in such organization. It says what follows if it is.

\begin{theorem}[ICCR preservation]\label{thm:preservation}
Let \(B\) and \(M\) be physical systems. Suppose \(M\models_{\ICCR}B\), and a consciousness-relevant property \(Q\) is grounded in intrinsic causal-computational organization in the sense of Principle \ref{pr:invariance}. Then
\[
Q(M)=Q(B).
\]
\end{theorem}

\begin{proof}
By \(M\models_{\ICCR}B\), there exist intrinsic partitions \(\Pi_B\) and \(\Pi_M\), together with structure-preserving bijections over states, inputs, interventions, and readouts, such that the mechanism-enriched transition structures and counterfactual response profiles of \(B\) and \(M\) are isomorphic. Hence
\[
\K(B,\Pi_B)\cong\K(M,\Pi_M).
\]
By Principle \ref{pr:invariance}, \(Q\) is an invariant of \(\K\) under such isomorphism. Therefore
\[
F_Q(\K(B,\Pi_B))=F_Q(\K(M,\Pi_M)),
\]
which implies \(Q(B)=Q(M)\). \qedhere
\end{proof}

This has an immediate implication for the simulation objection. Being a simulation of \(B\) does not by itself settle whether \(Q(M)=Q(B)\). As noted in Section~\ref{sec:iccr}, \(\Sim_R(M,B)\) is compatible with \(M\models_{\ICCR}B\): a system can be an external simulation of \(B\) and also realize its intrinsic causal-computational structure. Whenever \(M\) is, or can be built to be, such a realizer, Theorem~\ref{thm:preservation} applies and \(Q(M)=Q(B)\), provided \(Q\) satisfies ICF invariance. So the objection's inference from ``\(M\) is a simulation'' to ``\(M\) lacks \(Q\)'' is not valid in general: the right question is not whether \(M\) is called a simulation, but whether it can be taken to realize \(B\) under ICCR.

The conclusion is deliberately modest but important. The word ``simulation'' carries no automatic consciousness-negating force. To deny \(Q\) to \(M\), the critic must deny either the ICF grounding assumption or ICCR by specifying what consciousness-relevant intrinsic structure is missing.

\subsection{On the status of the preservation result}\label{sec:status}

Theorem~\ref{thm:preservation} follows almost immediately from the definitions, and this is by design. Once Principle~\ref{pr:invariance} is granted and ICCR is read as isomorphism of intrinsic causal-computational structure, preservation is just the statement that an invariant is preserved under isomorphism; the proof carries no hidden work. This is deliberate, and locating the content correctly matters for assessing what the paper claims. All of the substantive burden has been moved into two places: the assumption that the consciousness-relevant property is grounded in \(\K\) (Principle~\ref{pr:invariance}), and the definition of ICCR. The contribution of the paper is therefore not a theorem that forces a conclusion, but a precisification of what ``the same organization'' must mean. Thus, our present work should be reead as a reframing of the simulation question.

It is worth stating equally plainly what the result does not touch. ICCR is silent on the hard problem. Principle~\ref{pr:invariance} is a conditional about a grounding relation, not a claim that functional organization metaphysically necessitates phenomenal character. The framework is compatible with a reductive view on which \(\K\) constitutes \(Q\), and equally with a non-reductive view on which an organizational bridging principle links \(\K\) to \(Q\); on the latter reading the preservation result still holds, because it concerns only invariance under isomorphism of the grounding base. 

\subsection{Approximate preservation}\label{sec:approx}

Exact isomorphism is too strong for real physical systems. Biological brains are noisy, continuous, plastic, and partially stochastic. Artificial implementations may approximate rather than exactly duplicate their dynamics. The relevant comparison is therefore not exact sameness in every implementation detail, but preservation of the intrinsic organization at the appropriate grain. Differences below that grain do not alter the intervention-readout profile and should not count as differences in consciousness-relevant structure. Otherwise the realization relation would be too strict to apply to any realistic biological or artificial system, since no two physical systems agree in every microphysical detail. We therefore need an approximate version.

Let \(d_\K\) be a distance between intrinsic causal-computational structures. This may combine distances between transition kernels, intervention response profiles, dynamical invariants, information-flow structures, attractor landscapes, recurrence relations, or domain-specific quantities. Let \(d_Q\) be a distance or dissimilarity measure over consciousness-relevant properties, where such a measure is meaningful.

\begin{assumption}[Stability of consciousness-relevant grounding]
For a class of systems \(\mathcal{P}\), \(Q\) is stable under small changes in intrinsic causal-computational organization if there exists \(L>0\) such that
\[
d_Q(Q(P_1),Q(P_2))\leq L\,d_\K(\K(P_1),\K(P_2))
\]
for relevant \(P_1,P_2\in\mathcal{P}\).
\end{assumption}

\begin{theorem}[Approximate ICCR preservation]
If \(d_\K(\K(B),\K(M))<\eps\) and \(Q\) is stable with constant \(L\), then
\[
d_Q(Q(B),Q(M))<L\eps.
\]
\end{theorem}

\begin{proof}
Immediate from the stability assumption. \qedhere
\end{proof}

This result clarifies the role of digital approximation. A digital simulation need not reproduce every microphysical detail if consciousness-relevant organization is stable under the approximation and if the preserved structure includes the relevant intrinsic organization. Conversely, if small microphysical differences can produce large phenomenal differences, then the required \(d_\K\) must include those microphysical variables. The framework is neutral about this substantive issue; it requires the missing or preserved structure to be specified.

\section{Objections and replies}\label{sec:objections}
The objections to simulated consciousness are best understood not as a list of isolated complaints but as four families of challenges. Table \ref{tab:objection-roadmap} gives the roadmap. The replies below do not claim that every simulation is conscious. They show what each objection would have to establish in order to defeat ICCR.

\begin{table}[H]
\centering
\small
\begin{tabular}{p{0.20\linewidth}p{0.31\linewidth}p{0.39\linewidth}}
\toprule
\textbf{Family} & \textbf{Representative objections} & \textbf{ICF/ICCR response}\\
\midrule
Semantic and implementation objections & Simulation is not reality; water does not wet; Searle; Putnam-style triviality & Separate external description from intrinsic realization. Arbitrary labels, boundary-only profiles, and actual-trajectory matching do not count; physically supported intervention-readout profiles do.\\
\addlinespace
Counterfactual and absent-qualia objections & Maudlin's Olympia; Block's China Brain; absent, fading, and dancing qualia & Preserved counterfactual branches must be carried by components that are load-bearing in the system's actual dynamics, not by idle or merely appended machinery. Weird realizers are rejected when they fail ICCR, not merely because they are unfamiliar.\\
\addlinespace
Substrate and physical-realization objections & IIT; biological naturalism; thermodynamics; continuous dynamics; quantum or non-computable physics & If these structures are consciousness-relevant, they must be included in the target. The burden is to specify the missing structure, not merely invoke the word ``simulation.''\\
\addlinespace
Boundary and identity objections & Embodiment, situatedness, duplication, branching, epistemic indistinguishability & Expand the target boundary when the theory requires it. Conscious-process realization is distinct from personal identity and from mere indistinguishability.\\
\bottomrule
\end{tabular}
\caption{A compressed roadmap of the major objections and the corresponding ICCR strategy.}
\label{tab:objection-roadmap}
\end{table}

\subsection{Semantic and implementation objections}

The simplest objection says that a simulation of \(X\) is not \(X\). A simulated brain is not a biological brain, and a simulated body of water does not wet the external computer. These claims are true at the base-substrate level, but they do not decide the question at issue. ICCR does not infer consciousness from the fact that a system is described as a brain simulation. It asks whether the candidate system physically realizes the target's intrinsic causal-computational organization. If consciousness depends on biological chemistry as such, then ordinary digital simulation may fail. If consciousness depends on intrinsic causal-computational organization, then substrate identity is not the relevant standard. The slogan ``a simulation is not the real thing'' therefore equivocates between not being the same substance and not realizing the same consciousness-relevant structure.

Searle's Chinese Room makes a deeper point: syntax is not semantics, and formal rule-following does not by itself yield understanding \citep{Searle1980}. ICCR accepts this against bare formalism. It does not claim that a string of symbols, under an external interpretation, is conscious. It requires physically realized, mechanism-enriched causal organization. If the Chinese Room only manipulates symbols according to a rulebook while lacking the relevant intrinsic organization, it fails ICCR. If, however, the whole room were physically organized so as to realize the same mechanism-enriched intervention-readout structure as a genuinely understanding system, then the opponent would need to say what further structure is missing. The reply does not refute Searle by intuition. It relocates the burden from the word ``syntax'' to the specification of the relevant causal organization.

Putnam-style triviality arguments say that any physical system can implement any computation under a suitable mapping \citep{Putnam1988}. This is a serious problem for naive implementation. ICCR blocks it by refusing to count arbitrary mappings as realization. The partition must be physically supported, intervention-tracking, and invariant under relabelling. A rock may be externally mapped onto the state graph of a mind, but it will not thereby have the mechanism-enriched counterfactual response profile of that mind. This is the boundary-table case \(M_B\) of Section~\ref{sec:toy} in its limiting form: an externally imposed assignment from physical states to computational states has no law-governed counterpart to the target's interventions \(\doop(\cdot)\) unless the candidate is physically organized to update those very variables, and the admissibility criteria of Definition~\ref{def:intrinsic-partition} are exactly what an arbitrary Putnam mapping fails to satisfy. Triviality arguments succeed against tier-I label assignment. They do not establish that all tier-III dynamics-internal organization is observer-relative.

\subsection{Counterfactual and absent-qualia objections}

Maudlin's Olympia targets the use of counterfactual structure in computationalism \citep{Maudlin1989}. If consciousness depends on counterfactuals, why would inactive machinery that would have been used under different conditions matter to present consciousness? ICCR answers through its requirement that the partition be intrinsic (Definition~\ref{def:intrinsic-partition}). A counterfactual branch is preserved only if it is carried by a component that the intrinsic partition recognizes, and a component idle on the actual trajectory fails the intervention-tracking and relabelling clauses of that definition. Disconnected backup machinery, appended rulebooks, or physically idle devices therefore do not become part of the current intrinsic causal structure merely because an observer can describe what they would have done in a different setup. This aligns ICCR with dispositional and mechanistic responses to superfluous-structure problems \citep{Barnes1991,Klein2008}.

The toy model of Section~\ref{sec:toy} makes this criterion concrete and shows that it does not simply restate the dispositional reply. The worry Maudlin presses is that, along the actual trajectory, the conscious and the merely counterfactually-equipped systems are physically identical, differing only in idle machinery; so why should the idle part matter? The mechanism-enriched answer is that the relevant counterfactual branches are not carried by idle parts at all. In \(M_R\) the component that supports the preserved branch---the persistent register \(r_y\)---is read and written on the \emph{actual} run at every step, so the same physical variable whose perturbation \(\doop(r_y:=c)\) defines the counterfactual is dynamically load-bearing in generating the present state. Olympia's armature, by contrast, is precisely the kind of component that is idle on the actual trajectory: removing it leaves the present dynamics untouched, and by the relabelling- and intervention-tracking clauses of Definition~\ref{def:intrinsic-partition} it does not enter the intrinsic partition. The intrinsic-partition requirement therefore turns not on whether some machinery \emph{could} have responded, but on whether the component realizing each preserved branch is itself active in producing the current state---the very property that separates the recurrent realizer from an appended look-up or a dormant backup. This locates the dividing line between trace-replay triviality (too few counterfactuals) and Olympia (counterfactuals carried by inactive structure) in a single condition rather than in two ad hoc stipulations.

Block's China Brain and absent-qualia cases raise the worry that functionally equivalent but bizarre systems may lack experience \citep{Block1978}. The mechanism-enriched framework gives two replies. First, many bizarre cases fail the relevant conditions: they lack the temporal density, load-bearing counterfactual structure, intervention profiles, or internal readout structure of the biological target. Second, if a bizarre system really did satisfy ICCR for the correct theory, then rejecting its consciousness merely because it is bizarre would no longer be an argument. Chalmers's fading and dancing qualia arguments push in the same direction: if functional organization is preserved in fine-grained, causally integrated substitution, it becomes difficult to explain how phenomenal properties could change while all functional sensitivity to that change is absent \citep{Chalmers1995,Chalmers1996Mind}.

The mechanism-enriched construction also sharpens the lookup-table and unfolding debates. Boundary Canonical Functionalism treats a sufficiently expanded table or unfolded network as equivalent when the complete boundary profile is identical. ICCR does not automatically preserve that verdict. If the intrinsic target structure includes internal recurrent mechanisms, internal broadcasts, or intervention-sensitive self-models, then those structures must be preserved. A table or unfolded system passes only when it physically realizes the same mechanism-enriched structure. The result is neither an uncritical acceptance of tables nor a blanket rejection of them; it is a demand for structural specificity.

\subsection{Substrate and physical-realization objections}
The strongest substrate objection comes from IIT-style anti-functionalism. IIT holds that consciousness depends on a system's intrinsic cause-effect structure in physical terms, not on input-output or computational equivalence alone \citep{Oizumi2014,Albantakis2023,Findlay2024}. This is a serious challenge for any simple functionalist account: if a neural simulation fails to realize the intrinsic structure that IIT takes to be necessary for consciousness, then it fails to realize the corresponding experience.

IIT and ICF in fact share a motivation: both hold that the structure relevant to consciousness must be intrinsic to the system itself. The question is how to identify that structure. Some IIT-based arguments against simulated consciousness begin by classifying the candidate as a stored-program digital computer, and then assess its causal powers at the hardware level. From the perspective of ICF, this risks smuggling in an external standpoint: if the intrinsic perspective is taken seriously, the analysis should start not from the external category of units in a digital computer, but from the system's own states, transitions, and organization.

This reframing is itself a substantive position, not a neutral starting point, and it is worth being clear about where it disagrees with IIT. The disagreement is about which grain counts (Section~\ref{sec:grain}), and it has two parts.

The first is how the units themselves are fixed. IIT evaluates integrated information over partitions of a given set of units, and for a digital computer those units are taken to be the hardware elements. But tying the units to the hardware is itself an external choice. From the intrinsic perspective a unit need not be a physical element: it could be a combination of elements, such as a linear mixture of hardware variables, and such a basis may capture the system's own organization better than the hardware decomposition does. Which basis is the genuinely intrinsic one is left open, and IIT could in principle be applied over such non-standard units as well. The second part is exclusion. IIT's exclusion postulate says that a system specifies consciousness at only one grain, the maximally irreducible one. ICF does not assume this and maintains the possibility that a system can specify consciousness at more than one grain at once. In such a scenario, a simulation that realizes the consciousness-relevant organization at its own coarse grain would, on the ICF view, support experience at that grain, whether or not a finer grain does too. Our aim is not to claim that exclusion is wrong, but whether simulation can realize consciousness depends on the veracity of the exclusion axiom. 

Biological naturalism makes a related claim: consciousness may require biological causal powers of neurons, glia, metabolism, membranes, neurotransmitters, or organismic self-maintenance \citep{Seth2025}. ICCR again makes the issue explicit. If those biological features are constitutive, they belong in the target structure. A digital simulation lacking them fails. But the proponent of substrate dependence must identify the consciousness-relevant causal role of the substrate. Mere appeal to biological material as such does not explain why that material is necessary. If the relevant role can be implemented in a different medium, substrate difference is not decisive. If it cannot, ICCR will say that the artificial system fails.

Continuous dynamics, thermodynamics, electromagnetic fields, quantum effects, and non-computable physics should be handled in the same way. They are possible constituents of the intrinsic structure, not automatic refutations of ICCR. If consciousness depends on precise continuous dynamics, the approximation metric \(d_\K\) must include the relevant continuous variables. If it depends on thermodynamic non-equilibrium, energy flow, or temporal continuity, those features must be part of the target. If it depends on non-computable physics, then ordinary digital computation may fail, while analog or other physical implementations remain to be evaluated. ICCR is substrate-general only to the extent that the relevant intrinsic structure is preserved. It is not substrate-blind.

\subsection{Boundary, approximation, and identity objections}

Embodied and situated approaches object that consciousness is not generated by a brain in isolation. It may depend on bodily regulation, affective loops, action, sensorimotor contingencies, and environmental embedding. ICCR can accept this without difficulty. The target is then not an isolated brain but the relevant agent-body-world system. A candidate realization must include the virtual or physical body, environment, temporal coupling, and control loops if they are part of the constitutive organization. This broadens the realization target; it does not support the simple simulation objection. If the mind is partly world-involving, as extended-mind arguments suggest, then the correct ICCR boundary is also world-involving \citep{ClarkChalmers1998}.

Duplication and branching raise a different issue. If two simulations realize the same structure, do they produce one consciousness or two? If two physically distinct systems realize the relevant intrinsic organization, there are two instances of the relevant process, just as similar biological brains are not one person. Pausing and resuming a simulation raise questions about temporal continuity; branching raises questions about future identity. These problems matter ethically and metaphysically, but they do not show that realization fails.

Finally, epistemic indistinguishability is not the foundation of the argument. A simulated scientist might be unable to tell whether she is in a base world or a virtual world, but indistinguishability alone does not imply phenomenal equivalence. ICCR relies on preservation of intrinsic causal-computational organization, not merely on indistinguishability from the inside. Internal indistinguishability is a useful intuition only when it tracks structural preservation. If two systems are indistinguishable in behavior and self-report but differ in a consciousness-relevant intrinsic structure, ICF allows them to differ phenomenally. The burden again is to identify that structure.

\section{What the argument establishes}\label{sec:doesnot}

The argument establishes a conditional preservation result and a dialectical burden shift. The conditional result is straightforward: if consciousness is grounded in intrinsic causal-computational organization, and if an artificial or simulated system realizes that organization under ICCR, then the system realizes the corresponding consciousness-relevant properties. The dialectical result is equally important. A critic who denies consciousness to a brain simulation cannot stop at the phrase ``simulation is not instantiation.'' That phrase is true of static descriptions, arbitrary semantic mappings, mere trace players, and boundary-output devices that fail to realize mechanism-enriched structure. It is not a general argument against physically realized artificial systems. The critic must identify a missing consciousness-relevant structure---biological, thermodynamic, electromagnetic, quantum, IIT-theoretic, embodied, or otherwise---and explain why that structure is required.

The argument is deliberately weaker than a full proof of substrate independence. It does not show that consciousness is in fact computationally constituted. It does not show that ordinary digital computers already satisfy ICCR for human-like consciousness. It does not make boundary input-output equivalence sufficient, and it explicitly excludes movies, static descriptions, arbitrary semantic mappings, and trajectory replays. It does not show that biological details are irrelevant, that IIT is false, or that personal identity is preserved under uploading or copying. These limitations are not concessions to the simulation objection. ICF aims to identify the correct realization relation first, and then ask which physical systems satisfy it.

The most important clarification concerns the relation to Canonical Functionalism. The earlier boundary-canonical theorem remains valid. If two systems have the same complete input-output profile over a specified interface, their boundary canonical structures are isomorphic. What the present paper adds is that this may not be the canonical structure relevant to consciousness. If consciousness depends on internal mechanism, the correct canonical object must include mechanism-enriched intervention-readout profiles. The table and unfolding results therefore become diagnostic rather than decisive. They show what follows from a boundary-level functionalism. The present theory explains how to go beyond that boundary level without falling back into observer-relative mapmaking.

The resulting thesis is that the simple simulation objection is not a general objection to artificial consciousness. It succeeds only when the candidate system fails to realize the relevant intrinsic causal-computational organization. Once the relevant structure is specified, the dispute becomes substantive: which variables matter, which interventions reveal them, which readouts define the mechanism, which boundaries are adequate, what approximation tolerances are allowed, and which physical media can instantiate the required organization?

\section{Discussion}\label{sec:discussion}

\subsection{Computational correlates and organizational invariance}

The present framework can also be understood in relation to two closely connected ideas in the philosophy and science of consciousness: computational correlates of consciousness and organizational invariance. 

The search for neural correlates of consciousness asks which neural mechanisms are minimally sufficient for a conscious experience \citep{CrickKoch1990,Chalmers2000NCC}. On the other hand, a computational correlate of consciousness asks a related but more abstract question: which computational structures or processes correspond to conscious experience \citep{Cleeremans2005CCC}. This shift is important because the same consciousness-relevant organization might, in principle, be realized in systems that do not share the same biological substrate.

ICF agrees that the search for computational correlates is necessary, but it adds a constraint. A computational correlate of consciousness cannot be merely a model that an observer uses to describe a system, nor a pattern that matches behavior at the system boundary. If it is to play a constitutive role, it must be intrinsically realized by the system itself. In this respect, ICCR can be viewed as an attempt to specify when a computational correlate is not merely correlational or descriptive but a candidate intrinsic realizer of consciousness. The relevant structure must be grounded in the system's own state distinctions, internal organization, intervention profiles, and counterfactual dependencies.

This also clarifies the relation to Chalmers's principle of organizational invariance. Chalmers argues that conscious experience is preserved across systems with the same functional organization \citep{Chalmers1996Mind,Chalmers1995}. The present paper is sympathetic to this idea, but it refines the notion of ``same organization.'' Boundary-level input-output equivalence is too weak, because lookup tables, unfolded systems, and recurrent systems may share the same complete boundary behavior while differing in their internal mechanisms. Conversely, microphysical duplication is too strong, because ICF does not assume that every biological or microphysical detail is consciousness-relevant. The target is instead the intrinsic functional organization realized by the system: the internal structure that supports the relevant state distinctions, transitions, interventions, and readouts.

In this interpretation, the ICCR provides an intrinsic version of organizational invariance. If two systems realize the same intrinsic organization relevant to consciousness, then they should preserve the same properties relevant to consciousness, regardless of whether one is biological and the other artificial or simulated. However, equivalence must be assessed at the level of the realized internal organization, not merely at the level of external behavior. This is the sense in which the present account both inherits and strengthens the organizational-invariance tradition: it keeps the idea that organization matters, while adding the ICF requirement that the relevant organization must be system-intrinsic rather than observer-imposed.

Here, a clarification on novelty is warranted. The formal machinery underlying the mechanism-enriched construction is not itself new. Minimal-state and canonical-form constructions are standard in automata theory; the identification of states by their conditional futures, together with the associated minimality and uniqueness results, is the substance of computational mechanics and the $\varepsilon$-machine \citep{ShaliziCrutchfield2001,Crutchfield2012}. The use of interventions to individuate causal variables is a common approach in the causal-modelling literature \citep{Pearl2009} and is central to integrated information theory's cause-effect analysis \citep{Oizumi2014,Albantakis2023}. Read as formal results, the quotients defined above add little to these traditions. The contribution of this paper is to show that these tools, developed for the analysis of computation and causal structure, are exactly what is needed to make computational functionalism about consciousness more precise---and, in particular, to adjudicate the question of consciousness in artificial systems. What has been missing is the recognition that a mathematical description of a system's causal dynamics can itself serve to identify a physically realized intrinsic organization, without relying on an externally imposed interpretation. Bringing this clarity to the simulation debate, and to computational functionalism about AI consciousness more broadly, is the work the present paper aims to do.

\subsection{Implications for ICF tiers and a research program}\label{sec:program}

The formal definitions above are constitutive rather than methodological. They do not say that a theory of consciousness creates the relevant structure by selecting variables, readouts, or interventions. If ICF is correct, the consciousness-relevant structure should be found in a tier-III, dynamics-internal feature of the system itself. It is grounded in system-intrinsic instantiation (C1) and causal-dynamical organization under intervention (C2). In this sense, the mechanism-enriched canonical structure is the ideal target: an observer-independent structure fixed by the system's own organization, not by our descriptions of it.

This ideal does not remove the practical role of theory. For real brains and sophisticated artificial agents, we cannot measure, perturb, or model every physical degree of freedom. We do not know in advance which variables, boundaries, interventions, readouts, or grains will reveal the intrinsic structure relevant to consciousness. Scientific theories therefore have to tackle this epistemic challenge and can guide the search for the tier-III target. In particular, the grain-selection problem discussed in Section~\ref{sec:grain} is exactly the kind of question theory that must help to adjudicate empirically.

Tier-II constructs are useful in this practical sense. Global Workspace Theory, for example, proposes candidate variables and relations involving competition, ignition, broadcast, and global availability \citep{Baars1988,Dehaene2017}. Higher-order approaches propose monitoring and self-representation. Recurrent processing approaches propose local or global recurrent integration. Predictive and active-inference approaches propose generative models, precision control, and embodied action. Indicator-based approaches to AI consciousness can also be interpreted as theory-plural attempts to specify where to look for relevant structures \citep{Butlin2023,Butlin2026}. These constructs are not themselves the final intrinsic structure. Instead, they are disciplined hypotheses about which interventions and readouts may help reveal it.

The mechanism-enriched canonical framework offers a way to operationalize this search while preserving the ICF distinction between epistemology and constitution. For instance, a global workspace hypothesis may lead us to perturb candidate workspace variables and read out effects on specialized modules, report, memory, and flexible control. A recurrent processing hypothesis may lead us to disrupt recurrent pathways and measure changes in integration, stability, and temporal continuity. If the predicted intervention-readout relations fail, the tier-II construct is wrong or incomplete. If they succeed across subjects, perturbations, substrates, and tasks, the construct begins to approximate the tier-III dynamics-internal grain.

The framework therefore occupies a middle position where it proposes a transition from practical empirical research to an idealized mathematical target. Here, the ideal target is a mechanism-enriched canonical structure derived from the system's own dynamics, whereas the practical route uses tier-II theoretical constructs to approximate that target. If consciousness science eventually identifies organizational invariants, ICCR provides the realization relation needed to ask whether biological, artificial, or simulated systems instantiate them.

\subsection{The grain-selection problem}\label{sec:grain}
The research program just sketched runs into one open problem that deserves to be stated plainly, because it is where the present framework is least settled. An intrinsic partition, in the sense of Definition~\ref{def:intrinsic-partition}, is required to track stable, causally effective distinctions in the system's own dynamics; but these criteria do not explicitly pick out a single scale. A system can admit several intrinsic partitions at different grains, and since \(\K(P,\Pi_P)\) depends on the partition, the claim that \(M\) and \(B\) share the same intrinsic organization can be considered only if the relevant grain is fixed. This grain-selection problem was not addressed in this paper.

As we saw in discussing the IIT-based objection (Section~\ref{sec:objections}), this is exactly where ICF and integrated information theory part ways. IIT answers the grain problem with a built-in principle: its exclusion postulate holds that consciousness occurs at the spatiotemporal scale at which integrated information is maximal \citep{Oizumi2014,Albantakis2023}. ICF as developed here has no such maximization principle, though the framework could be extended to include one. Several intrinsic criteria could play that role---integrated information, effective information, or mechanistic individuation under intervention. Although both ICF and IIT agree that the consciousness-relevant grain must be defined as an intrinsic property of the system, what is the right criterion for selecting a particular grain remains an open question. 

However, our main results are not affected by the grain selection problem. Our central claims are conditional: once a grain is fixed, ICCR asks whether \(B\) and \(M\) share the same intrinsic structure at that grain, and the preservation theorem concludes that they then share the same consciousness-relevant properties. Which grain is the right one is not a premise of these claims, so leaving it open does not weaken them---provided only that the grain is constrained by the system's own dynamics rather than chosen by an external interpreter, which is what the criteria of Definition~\ref{def:intrinsic-partition} secure. Settling which grain is consciousness-relevant is then a task for the empirical program of the previous section, not a gap that the conditional argument must close in advance.

\subsection{From intrinsic realization to structure extraction}

The present framework identifies what must be preserved for simulated consciousness, but it does not yet provide a full theory of how to extract the relevant intrinsic structure from an arbitrary physical system. This is an important limitation. IIT, for example, is more developed in this respect: it offers a formal recipe for deriving cause-effect structures and evaluating their integration. ICCR should therefore not be presented as a rival measure of consciousness. Its role is to specify a realization constraint. It says that a consciousness-relevant structure, whatever its final formal characterization, must be intrinsic to the system, must include internal mechanisms rather than boundary behavior alone, and must be preserved under the relevant counterfactual patterns of intervention and readout.

This limitation also suggests the next step. An ICF-based theory of consciousness should develop methods for extracting effective intrinsic structures from complex systems. Such methods would identify stable state distinctions in the system's own dynamics, construct intervention-readout profiles over internal and boundary variables, and search for coarse-grainings that preserve these profiles while discarding irrelevant microphysical variation. The goal is not to impose a theory-dependent partition from the outside, but to discover dynamically constrained structures that the system itself realizes. In this sense, the problem of consciousness becomes partly a problem of intrinsic structure extraction.

Several criteria may guide this development. A candidate structure should be stable enough to support state distinctions, causally effective enough to constrain future evolution, robust enough to survive noise and perturbation, and rich enough to account for the phenomena associated with consciousness, such as integration, global availability, temporal unity, reportability, and self-monitoring. These criteria do not yet amount to a complete measure. They define a research program: to combine the generality of ICCR with the formal depth of IIT-like intrinsic-structure analysis.

\section{Conclusion}

The simulation objection begins from a correct observation: a simulation is not automatically an instantiation of the thing it simulates. A static description of a brain is not conscious. A movie of a brain is not conscious. A lookup table with arbitrary labels is not conscious merely because an observer maps it onto a mind. A boundary-output device that preserves only external behavior may fail to realize the internal mechanism relevant to consciousness. In each of these cases the objection is right. Its error is to generalize from them: none of these cases shows that a physically realized artificial system cannot instantiate the intrinsic causal-computational organization that grounds consciousness.

Intrinsic Computational Functionalism gives the right framework for this issue. It rejects naive computationalism by requiring system-intrinsic instantiation and causal-dynamical organization under intervention. It also rejects naive anti-computationalism by distinguishing observer-relative maps from observer-independent structures. Canonical Functionalism supplied a mathematically precise boundary case of this program: a canonical quotient over complete future input-output roles. The present paper extends that idea to mechanism-enriched canonical structure, where internal interventions and readouts are included in the target.

Once this distinction is in place, the simulation question becomes precise:
\begin{quote}
Does the candidate system realize the consciousness-relevant intrinsic causal-computational organization, including the relevant internal mechanisms, counterfactual intervention-readout profiles, and boundary-adequate agent-world structure?
\end{quote}

If yes, then substrate difference alone is not a reason to deny consciousness. If no, then the simulation fails. The word ``simulation'' does not decide the matter.

The resulting view is neither simple functionalism nor biological chauvinism. It is a constrained, intrinsic, mathematical functionalism. It turns the debate from metaphor to mathematical structure: not ``simulated water does not wet,'' but ``which causal-computational organization is realized, at which intrinsic grain, under which interventions and readouts, and with which invariants?'' That is the question a theory of simulated consciousness must answer.

\bibliographystyle{plainnat}
\bibliography{icf_simulation_realization_rev15}

\end{document}